\documentclass{llncs}

\usepackage{amssymb}
\usepackage{amsmath}
\usepackage{textcomp}
\usepackage{array}
\usepackage{lineno}
\usepackage{color}
\usepackage{comment}

\usepackage{graphicx}

\newcommand{\ket}[1]{|#1\rangle}
\newcommand{\braket}[2]{\langle #1|#2\rangle}
\newcommand{\Cent}[0]{\mbox{\textcent}}
\newcommand{\dollar}[0]{\$}

\newtheorem{fact}{Fact}

\title{Debates with small transparent quantum verifiers\thanks{A preliminary version appeared as ``A.~Yakary{\i}lmaz, A.~C.~C. Say and H.~G. Demirci, Debates with small   transparent quantum verifiers, {\em Developments in Language Theory - 18th International Conference, {DLT} 2014, Ekaterinburg, Russia, August 26-29, 2014. Proceedings\/},  {\em LNCS 8633}, (Springer, 2014), pp. 327--338.'' \protect\cite{YSD14}. The arXiv number is 1405.1655.}}

\author{Abuzer Yakary{\i}lmaz\inst{1}$^,$\thanks{Yakary{\i}lmaz was partially supported by CAPES with grant 88881.030338/2013-01, ERC Advanced Grant MQC, and B\"UVAK.}, A. C. Cem Say\inst{2}, H. G{\"{o}}kalp Demirci\inst{3}}

\institute{
National Laboratory for Scientific Computing, Petr\'{o}polis, RJ, 25651-075, Brazil
\and 
Bo\u{g}azi\c{c}i University, Department of Computer Engineering, Bebek 34342 \.{I}stanbul, Turkey
\and
University of Chicago, Department of Computer Science, Chicago, IL 60637 USA
\email{abuzer@lncc.br, say@boun.edu.tr, demirci@cs.uchicago.edu}
}

\authorrunning{Demirci, Say, and Yakary{\i}lmaz} 

\begin{document}

\maketitle

\begin{abstract}
We study a model where two opposing provers debate over the membership status of a given string in a language, trying to convince a weak verifier whose coins are visible to all. We show that the incorporation of just two qubits to an otherwise classical constant-space verifier raises the class of debatable languages from at most $\mathsf{NP}$ to the collection of all Turing-decidable languages (recursive languages). When the verifier is further constrained to make the correct decision with probability 1, the corresponding class goes up from the regular languages up to at least $\mathsf{E}$. We also show that the quantum model outperforms its classical counterpart when restricted to run in polynomial time, and demonstrate some non-context-free languages which have such short debates with quantum verifiers.
\end{abstract}

\section{Introduction}
It is well known that the model of alternating computation is equivalent to a setup where two opposing debaters (a prover and a refuter) try to convince a resource-bounded deterministic verifier about whether a given input string is in the language under consideration or not \cite{CKS81}. Variants of this model where the verifier is probabilistic, and the communications between the debaters   are restricted in several different ways, have been studied \cite{CFLS93,FK97,DSY15}. Quantum refereed games, where the messages exchanged between the debaters are quantum states, were examined by Gutoski and Watrous \cite{GW07}.

Most of the work cited above model the verifier as opaque, in the sense that the outcomes of its coin throws are not visible to the debaters, who have a correspondingly incomplete picture of its internal state during the debate. These models can therefore be classified as generalizations of private-coin interactive proof systems \cite{GMR89} to the competing multiple provers case. In this paper, we focus on models where all of the verifier's coins, as well as all communications, are publicly visible to all parties, making them generalizations of Arthur-Merlin games \cite{Bab85}. A recent result \cite{Yak12E} established that a very small quantum component is sufficient to expand the computational power of classical proof systems of this kind considerably, by studying a setup where an otherwise classical constant-space verifier is augmented by a quantum register of just two qubits.  We modify that protocol to show that the addition of a two-qubit quantum register to the classical finite state verifier raises the class of debatable languages from at most $\mathsf{NP}$ to that of all Turing-decidable languages. We also study the case where the verifier is required to take the correct decision with probability 1. We show that small quantum verifiers outperform their classical counterparts in this respect as well, exhibiting an increase from the class of regular languages to at least $\mathsf{E} = \mathsf{DTIME(2^{O(n)})}$. (Note that $\mathsf{E}$ is a proper subset of $\mathsf{EXP} = \mathsf{DTIME(2^{poly(n)})}$.) When we allow ourselves to use a richer set of transition amplitudes, similar results can also be shown for polynomial-time debates, and we demonstrate several non-context-free unary languages which have such fast quantum verifiers.

The rest of this paper is structured as follows: Section \ref{section:prel} describes our model and reviews previous work. Our result on the computational power of the model with a two-qubit constant-space verifier in the two-sided bounded error case is presented in Section \ref{section:decidable}. Section \ref{section:zeroerror} contains an examination of the more restricted zero-error case. We establish the superiority of polynomial-time quantum verifiers over their classical counterparts, and give several tally languages handled by such machines, in Section \ref{section:poly}. Section \ref{section:conc} concludes the paper with some remarks on the possible usage of multihead automata as verifiers, and an open question.

\section{Preliminaries} \label{section:prel}
Consider an interactive system consisting of three actors: two debaters, namely a Prover named Player 1 (P1)  and a Refuter named Player 0 (P0), respectively, and a computational agent called the Verifier (V). All actors have access to a common input string $w$. P1 tries to convince V that $w$ is a member of the language $L$ under consideration, whereas P0 wants to make V reject $w$ as a non-member. The debaters communicate with each other and the verifier through a communication cell which is seen by every actor. V  takes this communication and the outcomes of its coin into account for each step of its computation. The debate continues in this way until the computation of V is terminated as it reaches a decision. We assume that both debaters also see the coin outcomes of V as they occur, and thereby have complete information about the state of the verifier at any point.

In such a setup, we say that language $L$ \textit{has debates checkable by a machine V with error bound} $\epsilon \in [0,\frac{1}{2})$ if 
\begin{itemize}
\item for each $w \in L$, P1 is able to make V accept $w$ with probability at least $1-\epsilon$, no matter what P0 says in return,
\item for each $w \notin L$, P0 is able to make V reject $w$ with probability at least $1-\epsilon$, no matter what P1 says in return.
\end{itemize}
A language is said to be \textit{debatable} if it has debates checkable by some verifier. Note that the class of debatable languages is closed under complementation.

We focus on verifiers which are only allowed to operate under constant space bounds. When V is set to be a deterministic two-way finite automaton, the system described above is equivalent to an alternating two-way finite automaton, and the class of debatable languages coincides with the regular languages \cite{LLS78}. When one replaces V with a two-way probabilistic finite automaton, one obtains a setup equivalent to the alternating probabilistic finite automata (2apfa's) defined in \cite{CHPW98}, and the issue of whether irrational numbers are  allowed as transition probabilities becomes significant. In this paper, we start with both classical and quantum models defined in the most general (reasonable\footnote{When uncomputable real amplitudes are allowed, every language has debates (in which the verifier listens to only one of the players) that are checkable in double exponential time by constant-space quantum verifiers  \cite{SY14B}.}) way by allowing computable real transition probabilities and amplitudes, and we make sure to be fair by applying any restrictions in this regard simultaneously to both versions while comparing them under different conditions. 
\begin{figure}[!th]
	\centering
	\fbox{
	\includegraphics[scale=0.4]{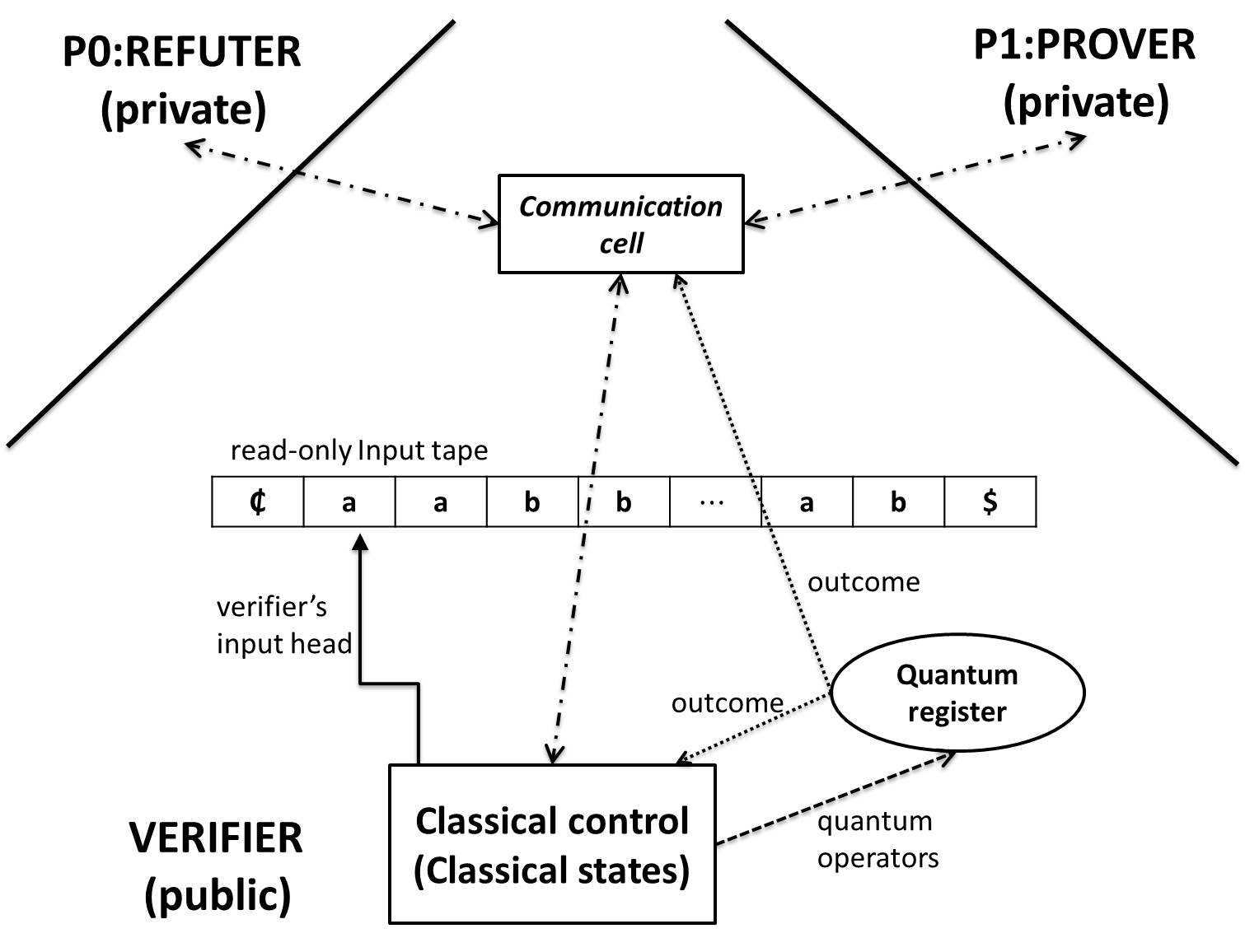}
	}
	\caption{The details of our debate system}
	\label{fig:debater}
\end{figure}
The public-coin quantum verifier model (see Figure \ref{fig:debater}) that we will use is the two-way finite automaton with quantum and classical states (2qcfa) \cite{AW02}, in which the quantum and classical memories are nicely separated, allowing a precise quantification of the amount of ``quantumness" required for our task:\footnote{The 2qcfa definition we present is a ``modernized" version of the original model in \cite{AW02}, generalizing and simplifying the quantum transition setup by using superoperators; see \cite{YS11A}. The two versions are equal in computational power.}

 Formally, a 2qcfa verifier $V$ is an $ 8 $-tuple
\[
	V = (Q,S,\Sigma,\Gamma,\delta,q_1,s_1,s_a,s_r),
\]
where $ Q $ is the set of quantum states, $ S $ is the set of classical states, $ \Sigma $ is the input alphabet, $\Gamma$ is the communication alphabet, $ \delta  $ is the set of transition functions to be described below, $ q_1 \in Q $ is the initial quantum state, $s_1 \in S$ is the  initial classical state, and $ s_a \in S $ and $ s_r \in S $ are respectively the  accepting and rejecting states, such that $ s_a \neq s_r$.  Any given input $ w \in \Sigma^* $ is placed on a read-only single-head input tape between the left- and right-endmarkers ($\Cent$ and $\dollar$, respectively). Let $ \tilde{\Sigma} = \Sigma \cup \{  \Cent,\dollar\} $. 

At the beginning of the computation, the classical and quantum parts of $ V $ are in states $ s_1 $ and $ q_1 $, respectively, and the head is placed on the left end-marker. Computation halts and the input is accepted (rejected) when $ V $ enters $s_a$ ($s_r$). The set of non-halting classical states is subdivided to the sets of reading states ($ S_r $) and communication states ($S_c$), i.e. $ S = S_r \cup S_c \cup \{ s_a,s_r \} $. In each step, the automaton either communicates with the debaters or makes local updates, each such local update consisting of a quantum transition, followed by  a classical one. Therefore, the ``program" $\delta$ is a set whose elements are the three transition functions $ \delta_c $, $ \delta_q $, and $ \delta_s $ that are respectively responsible for the classical communications, local quantum transitions, and local classical transitions, as follows.


When $V$ is in a communication state, its next move is determined by the function $\delta_c$. For $s_c \in S_c$, $ \sigma \in \tilde{\Sigma} $, $ s_r \in S_r $, and $\gamma \in \Gamma$, the function value $ \delta_c(s_c,\sigma) = (s_r,\gamma) $ means that $ V $ will  write the symbol $\gamma$ to the communication cell and switch to the  reading state $s_r$ upon scanning the symbol $\sigma$ on the input tape when originally in communication state $s_c$. After $V$ makes its transmission in this manner, the prover and the refuter emit their responses, say, the symbols $ \gamma_p \in \Gamma $ and $ \gamma_r \in \Gamma $, respectively,  by writing them to their slots in the communication cell. 

Upon entering a reading state, $V$ performs two transitions, the first dictated by the function $\delta_q$, and the second by $\delta_s$. 
For $s_r \in S_r$, $ \sigma \in \tilde{\Sigma} $, and $ \gamma_p, \gamma_r \in \Gamma $, the function value $ \delta_q(s_r,\sigma,\gamma_p,\gamma_r) = \mathcal{E}$ directs $ V $ to apply the superoperator (Figure \ref{fig:superoperators}) $\mathcal{E}$ to  its quantum register if it scans the symbol $\sigma$ on the input tape and the communication symbols $ \gamma_p$ and $ \gamma_r$ in the slots of P1 and P0 in the communication cell when it is in reading state $s_r$. This results in a measurement outcome $i$ to be produced and sent to all actors automatically. The function $\delta_s$ is then used to determine the next classical state and head position: For $ d \in \{ left, stay\mbox{-}put,right \} $, the function value $ \delta_s(s_r,\sigma,\gamma_p,\gamma_r,i) = (s,d)$ causes $V$ to switch to state $ s \in S $, and move the input head in direction $d$.


\begin{figure}[!ht]
	\centering
	\small
	\fbox{
	\begin{minipage}{0.97\textwidth}
		For a 2qcfa with $j$ quantum states, each superoperator $ \mathcal{E} $ is composed of a finite number of $j\times j$ matrices called \textit{operation elements},
		$ \mathcal{E} = \{ E_{1}, \ldots, E_{k} \} $, satisfying
		\begin{equation}
		\label{eq:completeness}
			\sum_{i=1}^{k} E_{i}^{\dagger} E_{i} = I,
		\end{equation}
		where $ k \in \mathbb{Z}^{+} $, and the indices are the measurement outcomes.
		When a superoperator $\mathcal{E}$ is applied to 
		a quantum register in state $\ket{\psi}$, then
		we obtain the measurement outcome $i$ with probability 
		$ p_{i} = \braket{\widetilde{\psi_{i}}}{\widetilde{\psi_{i}}} $,
		where $\ket{\widetilde{\psi_{i}}}$ 
		is calculated as $ \ket{\widetilde{\psi}_{i}} = E_{i} \ket{\psi} $, and $1 \leq i \leq k$.
		If the outcome $i$ is observed ($p_{i} > 0 $), the new state of the system 
		is obtained by normalizing $ \ket{\widetilde{\psi}_{i}} $
		which is $ \ket{\psi_{i}} = \frac{\ket{\widetilde{\psi_{i}}}}{\sqrt{p_{i}}} $.
		Moreover, the quantum register can be set to a predefined quantum state by an initialize operator with a single outcome. 
	\end{minipage}
	}
	\label{fig:superoperators}
	\caption{Superoperators (adapted from \protect\cite{Yak12E})}
\end{figure}


One obtains the definition of the quantum Arthur-Merlin (qAM) systems of \cite{Yak12E} when one removes P0 from the picture described above.  Some of our results  on  small quantum verifiers for debates are based on the following result: 
\begin{fact}
\label{fact:qAMrecognizable}
For any error bound $\epsilon>0$, every Turing-recognizable language (recursively enumerable language) has an Arthur-Merlin system where the verifier uses just two quantum bits, (i.e. four quantum states,) all transition amplitudes are rational numbers, members of the language are accepted with probability 1,   nonmembers are accepted with a probability not greater than $\epsilon$, and a dishonest P1 can cause the machine to run forever without reaching a decision.
\end{fact}
\begin{proof}
We outline the basic idea, referring the reader to \cite{Yak12E} for a detailed exposition of this proof, which also shows how the verifier in question can be implemented with only rational number entries in its quantum operators.  Let $T$ be the single-tape Turing machine recognizing the language $L$ under consideration. For any input string $w$, P1 (the only debater in this restricted scenario) is supposed to send the computation history (i.e. the sequence of configurations) of $T$ on input $w$ to the verifier V. Some of the possible outcomes of V's observations of its quantum register will be interpreted as  ``restart" commands to P1. At any point, V may interrupt P1 and ask it to restart sending the computation history from the beginning in this manner. (In fact, the verifier is highly likely to require a restart at each step.) 

Whenever the verifier catches P1 lying (i.e. giving an incorrect configuration description), it rejects the input. If the verifier reads a computation history sent by P1 all the way to its completion by a halting configuration without detecting an incorrect configuration,   it halts with a decision paralleling the one described in that history with a certain non-zero probability, and requests a restart with the remaining probability.

A classical public-coin finite automaton faced with this task would not be able to compare two consecutive configuration descriptions $c_i$ and $c_{i+1}$ (which may be very long).\footnote{The associated complexity class for machines with rational transition amplitudes is known \cite{CL88} to be included in $\mathsf{P}$. When the verifier is allowed to hide its coins, its power increases \cite{CL89}.} A 2qcfa verifier handles this problem by encoding the substrings in question  into the amplitudes of  its quantum states.\footnote{Actually, this encoding can also be performed by a classical probabilistic machine \cite{Rab63}. It is the subsequent subtraction that is impossible for classical automata.}  Let $next(c)$ denote the description of the configuration that is 
the legitimate successor, according to the transition function of $T$, of configuration $c$, and let $e(x)$ denote an integer that encodes string $x$ according to a technique to be described later. After 
 the description  $c_{i+1}$ has been read, the amplitudes of the  quantum states of V form the vector
$
	\alpha\left(1~~ e(next(c_{i})) ~~ e(c_{i+1}) ~~ e(next(c_{i+1})) \right)^\mathsf{T},
$
where $\alpha$ is a small rational number.
(The amplitude of the first state is used as an auxiliary value during the encoding \cite{Yak12E}, as will also be seen in the next section.)

 When P1 concludes the presentation of $c_{i+1}$, V executes a move that has the effect of subtracting $\alpha e(next(c_{i}))$ from  $\alpha e(c_{i+1})$,  rejecting with a probability equal to the square of the difference, continuing with some little probability after placing the encoding of $next(c_{i+1})$  into the second state's amplitude and resetting the third and fourth amplitudes to zero for beginning the next encode-compare stage, and requesting a restart with the remaining probability. If $c_{i+1}$'s description is indeed equal to the valid successor of $c_{i}$, the subtraction mentioned above yields zero probability of rejection. Otherwise, the rejection probability arising from a transition error within a computation history is guaranteed to be a big multiple of the acceptance probability that may arise due to that spurious history ending with an accepting configuration.  

If $w \in L$, P1 need only obey the protocol, sending the accepting computation history, restarting each time V tells it to do so. In each try, P1 has a small but nonzero probability of sending the full history without being interrupted, leading to a nonzero probability of halting with acceptance. Since V will detect no transition errors between configurations, the probability of rejection is zero.

If $w \notin L$, any attempt of P1 to trick V to accept $w$ with high probability by sneaking a transition error to the history and ending it with an accepting configuration will be foiled, since the rejection probability associated with the defect in the history can be guaranteed to be as big a multiple of the final acceptance probability as one desires. There is, however, one annoyance that P1 can cause V in this case: If P1 sends an infinite-length ``configuration description" at any point\footnote{Except at the beginning, since V can check the first configuration itself by matching it with the input.} during its presentation, V will never reach the point where it compares the two amplitudes it uses for encoding, and it will therefore fail to halt.
\end{proof}

\section{Small transparent verifiers for all decidable languages}\label{section:decidable}
Our first result is a generalization of the proof of Fact \ref{fact:qAMrecognizable} to the  setup with two debaters described in the previous section.\footnote{In separate work, the techniques of \cite{Yak12E} were used to define a model called q-alternation \cite{Yak13A}. This model is distinct from debate checking in the same sense that the two equivalent definitions of classical nondeterminism (the ``probabilistic machine with zero cut-point"   and the ``verifier-certificate" views) lead to quantum counterparts (\cite{ADH97} and \cite{KSV02}, respectively) which are remarkably different from each other.}

In this section, all entries of the quantum operators of the machines to be described are rational numbers, meaning that the probabilities of the outcomes are always rational.\footnote{The classical probabilistic finite automata, to which we compare our quantum model, can only flip fair coins. It is known that this is sufficient for two-way automata to realize any rational transition probability.} With rational transition probabilities, the class corresponding to the classical counterpart of this verifier model becomes one that should be denoted  $\mathsf{\forall BC\mbox{-}SPACE}(1)$ in the terminology of \cite{Co89}, and is known to contain some nonregular languages \cite{DS92}, and to be contained in  $\mathsf{NP}$. We will show that the addition of a small amount of quantum memory to the probabilistic model increases the power hugely, all the way to the class of decidable languages.
\begin{theorem}\label{theorem:decidable}
For every error bound $\epsilon>0$, every Turing-decidable language has debates checkable by  a 2qcfa with four quantum states, only rational entries in its quantum operators, and with error bounded by $\epsilon$.
\end{theorem}
\begin{proof}
We modify the verifier V described in the proof of Fact \ref{fact:qAMrecognizable} in Section \ref{section:prel} to obtain a new verifier V$_1$ as follows: V$_1$ listens to both P0 and P1 in parallel. In the protocol imposed by V$_1$, both debaters are expected to behave exactly as P1 was supposed to behave in that earlier proof; transmitting the computation history of the single-tape Turing machine $T$ for language $L$ on input string $w$, interrupting and restarting transmissions whenever V$_1$ observes an outcome associated with the ``restart" action in its quantum register.

The strategy of V$_1$ is based on the fact that the two debaters are bound to disagree at some point about the computation history of $T$ on $w$. As long as the same description is coming in from both debaters, V$_1$ uses  the same technique mentioned in the proof of Fact \ref{fact:qAMrecognizable}, to be described in more detail shortly, for encoding the successive configurations. At the first point within a history when a mismatch between the two debaters is detected, V$_1$ uses its register to flip a fair coin to choose to trace one or the other debater's transmission from that time. The chosen debater's description of what it purports to be the computation history is then checked exactly as in the earlier proof, and the other debater is ignored until a restart is issued by V$_1$ to both players during (or at the end of) that check. The truthful debater always obeys the protocol. In the  case that the other debater's transmission is identical to that of the truthful one, V$_1$ parallels the decision of $T$ depicted by both debaters.

If it sees the debater it is tracing violating the protocol, for instance, making a transition error, V$_1$ rules in favor of the other player.  When it sees a debater announcing the end of a computation history, V$_1$ decides in that debater's favor with some probability, and demands a restart with the remaining probability. Like the program described in the proof of Fact \ref{fact:qAMrecognizable}, V$_1$ is constructed so that the probability of the decision caused by the detection of a transition error in a computation history is guaranteed to be much greater than the probability of the decision caused by mimicking the result described at the end of that history.

A full description of V$_1$ would involve the complete presentation of its  transition functions, including all the operation elements of every superoperator. We will give a higher-level description of the program and its execution at a level that will allow the interested reader to construct the full 2qcfa if she wishes to do so.

 A segment of computation which begins with a (re)start, and ends with a halting or restarting configuration
will be called a ``round" \cite{YS10B}. In each such round, each debater is supposed to transmit a string of the form 
\[
 c_1 \dollar \dollar c_2  \dollar \dollar \cdots \dollar \dollar c_{h-1} \dollar \dollar,
\]
where $c_1$ is the description of the start configuration of $T$ on $w$, each $c_{i+1}$ is the legal successor of the corresponding $c_i$, and $c_{h-1}$ is the last configuration in the computation history before the halting configuration. (V$_1$ will be able to understand whether the successor of $c_{h-1}$ is an accepting or rejecting configuration by focusing on the symbols around the tape head in $c_{h-1}$.) We assume that each configuration description ends with the blank symbol \#, and that the alphabet $\Gamma$ used to write the configurations does not include the $\dollar$ symbol. Fix an ordering of the symbols in $\Gamma$, and let $e(\sigma)$ denote the position of any symbol $\sigma\in\Gamma$ in this ordering. Let $m$ be an integer greater than the cardinality of $\Gamma$, we will fix its value later.

The state of the quantum register is set to 
$
	\ket{\psi_{1,0}} = \left(  1~  0~  0~  0 	 \right)^\mathsf{T}
$
at the beginning of each round. 

Let $ l_{i} $ be the length of 	$ c_1 \dollar \dollar c_2  \dollar \dollar \cdots  c_{i} \dollar \dollar $ ($ i>0 $).

As it reads the string $ w_{1} = c_1 \dollar \dollar $ from the debaters, V$_1$ both compares it with the input to catch a debater that may lie at this point, and also
applies a superoperator corresponding to each symbol of $w_1$ to the register in order to encode $ next(c_1) $ as a number in base $m$ (times a factor that will be described later) into the amplitude of the second quantum state. One operation element of the superoperator $ \mathcal{E}_{1,j}$ applied when reading the $j$th symbol, say, $\sigma$, of $w_1$ is 
	\[ 
		E_{1,j,1}=\frac{1}{d} \left( 
			\begin{array}{cccr}
				1 & ~~0~~ & ~~0~~ & ~~0 \\ 
				e(\sigma) & m & 0 & 0 \\ 
				0 & 0 & 0 & 0 \\ 
				0 & 0 & 0 & 0
			\end{array} 
		\right),
	\]
where $d$ is an integer which has the properties to be described now.\footnote{Note that the ``names" we are using for the superoperators are based on their application position on the debater transmissions;  this same superoperator would be applied again (but would have a different index in our exposition) if another $\sigma$ comes up elsewhere in the transmission of $c_1$.}  Since $ \mathcal{E}_{1,j}$ would not obey the wellformedness criterion (Equation \ref{eq:completeness} in Figure \ref{fig:superoperators}) if its only operation element were $E_{1,j,1}$, we add as many $4\times 4$ rational matrices as necessary as \textit{auxiliary operation elements} of $ \mathcal{E}_{1,j}$ to complement its single \textit{main operation element} $E_{1,j,1}$ to ensure that Equation \ref{eq:completeness} is satisfied. Furthermore, we do this for all superoperators to be described in the rest of the program in such a way that  each of their main operation elements can be written with the same factor $\frac{1}{d}$ in front, as we just did for $E_{1,j,1}$. This is the property that $d$ must satisfy, and such a $d$ can be found easily \cite{Yak12E,YS11A}.

The observation outcome associated with all auxiliary operation elements will be interpreted as a ``restart" command to the debaters. Some operation elements to be described below are associated with halting (acceptance or rejection). The outcomes of all remaining operation elements, including the $E_{1,j,1}$, are ``continue" commands. 

Depending on whether the length of $T$'s configuration description increases as a result of its first move or not,  we have the following cases:
	\begin{itemize}
		\item  If $ |next(c_1)| = |c_1| $, the main operation elements of 
			$ \mathcal{E}_{1,|c_1|} $ and $ \mathcal{E}_{1,|c_1 \dollar|} $ are
			\[ 
				\frac{1}{d} \left( 
					\begin{array}{cccr}
						1 & ~~0~~ & ~~0~~ & ~~0 \\ 
						e(\#) & m & 0 & 0 \\ 
						0 & 0 & 0 & 0 \\ 
						0 & 0 & 0 & 0
					\end{array} 
				\right)
				\mbox{\normalsize and }
				\frac{1}{d} \left( 
				\begin{array}{cccr}
					1 & ~~0~~ & ~~0~~ & ~~0 \\ 
					0 & 1 & 0 & 0 \\ 
					0 & 0 & 0 & 0 \\ 
					0 & 0 & 0 & 0
				\end{array} 
				\right),
			\]
			respectively,
			since the encoding of $ next(c_1) $ is finished by superoperator 
			$ \mathcal{E}_{1,|c_1|} $. 
		\item If $ |next(c_1)| = |c_1|+1 $, and the $|c_1|$th symbol of $next(c_1)$ is $\sigma$, the main operation elements of 
			$ \mathcal{E}_{1,|c_1|} $ and $ \mathcal{E}_{1,|c_1 \dollar|} $ are
			\[
				\frac{1}{d} \left( 
					\begin{array}{cccr}
						1 & ~~0~~ & ~~0~~ & ~~0 \\ 
						e(\sigma) & m & 0 & 0 \\ 
						0 & 0 & 0 & 0 \\ 
						0 & 0 & 0 & 0
					\end{array} 
				\right)
				\mbox{\normalsize and }
				\frac{1}{d} \left( 
				\begin{array}{cccr}
					1 & ~~0~~ & ~~0~~ & ~~0 \\ 
					e(\#) & m & 0 & 0 \\ 
					0 & 0 & 0 & 0 \\ 
					0 & 0 & 0 & 0
				\end{array} 
				\right),
			\]
			respectively,
			since the encoding of $ next(c_1) $ is finished by  superoperator 
			$ \mathcal{E}_{1,|c_1|+1} $. 
	\end{itemize}
	The main operation element of $ \mathcal{E}_{1,|c_1 \dollar \dollar|} $  just multiplies the state vector by $\frac{1}{d}$.

 As long as the debaters are in agreement, and a halting configuration has not been detected,
 each configuration description block $ w_i = c_i \dollar \dollar $ ($ i \geq 2 $) is processed in the following manner.
	The state vector is 
	\[
		\ket{ \widetilde{ \psi_{i,0} } } = \left( \frac{1}{d} \right)^{l_{i-1}} 
				\left( 1 ~~ e(next(c_{i-1})) ~~ 0 ~~ 0 	 \right)^\mathsf{T}
	\]
	at the beginning of the processing.
	The tasks are:
	\begin{enumerate}
		\item To encode   $ c_i $ and $ next(c_i) $ into the amplitudes of the third and fourth quantum states, respectively, during the processing of the substring $ c_i \dollar $, and
		\item To accept (resp. reject) the input if $ next(c_{i}) $ is an accepting (resp. rejecting) configuration, or 
				to prepare for the $ (i+1)^{st} $ configuration description block if $ next(c_i) $ is not a halting configuration, during the processing of the final $ \dollar $
symbol.	
	\end{enumerate}	
	The details of superoperators to encode $ c_i $ and $ next(c_i) $ are similar to the ones given above.
	For each $ j \in \{ 1, \ldots, |c_i|-1 \} $, the main operation element of 
	$ \mathcal{E}_{i,j} $ is
	\[
		\frac{1}{d} \left( 
			\begin{array}{cccr}
				1 & ~~0~~ & ~~0~~ & ~~0 \\ 
				0 & 1 & 0 & 0 \\ 
				e(\sigma) & 0 & m & 0 \\ 
				e(\gamma) & 0 & 0 & m
			\end{array} 
		\right),
	\]
	where $\sigma$ and $\gamma$ are the $j$'th symbols of $c_i$ and $next(c_i)$, respectively.
	 $ \mathcal{E}_{i,|c_i|} $ and $ \mathcal{E}_{i,|c_i \dollar|} $ handle the two cases where $e(next(c_i))$ may or may not be longer than $e(c_i)$, similarly to the superoperators seen for the processing of $c_1$ \cite{Yak12E}.
	Thus, before applying $ \mathcal{E}_{i,|c_i\dollar\dollar|} $, the state vector becomes
	\begin{equation}
		\label{eq:state-vector-end-2}
		\ket{ \widetilde{ \psi_{i,|c_i \dollar|} } } = \left( \frac{1}{d} \right)^{l_i-1} 
		\left(  1 ~~ e(next(c_{i-1}) ~~ e(c_i) ~~ e(next(c_i)) 	 \right)^\mathsf{T}.
	\end{equation}
	
	Task (2) described above is to be realized by operator $ \mathcal{E}_{i,|c_i\dollar\dollar|} $, which has one main operation element, as described in Figure \ref{fig:bilmemkac}. 
\begin{figure}[!ht]
\centering
\footnotesize
\fbox{
\begin{minipage}{0.96\textwidth}
	\begin{tabular}{l|c}
		DESCRIPTION & OPERATOR
		\\
		\hline
		\begin{minipage}{0.75\textwidth} 
		\vspace*{2pt}
			If $ next(c_{i}) $ is a halting configuration, then  this operator is applied with the action of acceptance or rejection, as indicated by  $ next(c_{i}) $, associated with the outcome. The input is thereby accepted  or rejected with probability 
		$
			p_{1}=\left( \frac{1}{d} \right)^{2l_i}.
		$
		The round is terminated in this case.
		\vspace*{2pt}
		\end{minipage}
		\hspace{1pt}
		&
		$~
			\dfrac{1}{d} \left( 
				\begin{array}{ccrr}
					1 & ~~0~~ & 0 & ~~0 \\ 
					0 & 0 & 0 & 0 \\ 
					0 & 0 & 0 & 0 \\ 
					0 & 0 & 0 & 0
				\end{array} 
			\right)
		$
		\\
		\hline
		\begin{minipage}{0.75\textwidth}
			If $ next(c_{i}) $ is not a halting configuration, then this operator is applied. The state vector becomes 
		$
			\ket{ \widetilde{ \psi_{{i+1},0} } } = \left( \frac{1}{d} \right)^{l_i} 
			\left(  1 ~~e(next(c_{i})) ~~ 0 ~~ 0 	 \right)^\mathsf{T}.
		$
		\end{minipage}
		&
		\begin{minipage}{0.21\textwidth}
		\vspace{2pt}
		$
			\dfrac{1}{d} \left( 
				\begin{array}{ccrr}
					1 & ~~0~~ & 0 & ~~0 \\ 
					0 & 0 & 0 & 1 \\ 
					0 & 0 & 0 & 0 \\ 
					0 & 0 & 0 & 0
				\end{array} 
			\right)
		$
		\end{minipage}
	\end{tabular}
\end{minipage}
}
\caption{Operation element for preparing for the next configuration in the debater stream}
\label{fig:bilmemkac}
\end{figure}

After a disagreement between the debaters is noticed, the verifier picks a debater with probability $ \frac{1}{2} $. (A fair coin can be implemented in this setup by the superoperator $
		\mathcal{E} = \left\lbrace E_{h_{1}} = \frac{1}{2} I, E_{h_{2}} = \frac{1}{2} I,
		E_{t_{1}} = \frac{1}{2} I, E_{t_{2}} = \frac{1}{2}I \right\rbrace
	$
	with the outcomes for the first two operation elements interpreted as heads and the other ones as tails, for instance.)
The processing of the transmission of the chosen debater is the same as the processing of the common stream, except for the last superoperator dealing with the final $\dollar$ symbol of each description block. That superoperator has two main operation elements.
	The first one realizes the first actual transition correctness check:
	\[
		\frac{1}{d} \left( 
					\begin{array}{ccrr}
						0 & ~~0~~ & 0 & ~~0 \\ 
						0 & 1 & -1 & 0 \\ 
						0 & 0 & 0 & 0 \\ 
						0 & 0 & 0 & 0
					\end{array} 
				\right).
	\]
	The associated action of this operation element is to reject the claim of this debater.
	Therefore, when talking to P0 (resp., P1), the input is accepted (resp., rejected) with probability 
	$
		\left( \frac{1}{d} \right)^{2l_i} \left( e(next(c_{i-1})) - e(c_i) \right)^{2},
	$
	which is zero if the check succeeds ($ next(c_{i-1}) = c_i $), and
	is at least 
	$
		p_{2}=\left( \frac{1}{d} \right)^{2l_i} m^{2}
	$
	if the check fails ($ next(c_{i-1}) \neq c_i $).
	Since the last symbols of $ next(c_{i-1}) $ and $ c_i $ are identical,
	the value of $ | e(next(c_{i-1})) - e(c_i) | $ can not be less than $ m $ in this case.
	
	The second main operation element is the one already described in Figure \ref{fig:bilmemkac}, which either halts and decides, or readies the state vector for scanning the next configuration (with small probability) depending on whether that next configuration is a halting one or not. 
	
	Note that if the chosen debater is cheating and never sends any $\dollar$'s, then the communication with it terminates with probability 1 without any decision.

The overall acceptance probability of such a ``program with restart" equals the ratio of the acceptance probability to the halting probability in a single round \cite{YS10B}. The probability that the truthful debater will be selected after the disagreement is $\frac{1}{2}$. If this happens, V$_1$ will reach a halting state with the correct decision with some small probability, and restart with the remaining probability. In case the other debater is selected, there are two different possibilities of deception. If that debater presents an infinite ``configuration", V$_1$ will restart sooner or later. Otherwise, if a finite but spurious history with one or more incorrect transitions is presented, V$_1$ 
 may make the wrong decision with some small probability $p_1$, but this is more than compensated by the much greater probability $p_2$ of its making the correct decision earlier on, when the transition error(s) in this history were detected. Overall, the error rate of $\epsilon$ of V$_1$ is bounded by $\frac{p_1}{p_1+p_2}=\frac{1}{m^2+1}$, and can be tuned down to any desired positive value by choosing $m$, the base of the encoding used, to be a sufficiently large integer.
\end{proof}

\section{Debates with zero error}\label{section:zeroerror}

In classical computation, the benefits of using random bits come at the cost of incurring some nonzero probability of error; and ``zero-error" probabilistic finite automata can be shown trivially to be no more powerful than their deterministic counterparts. We will now show that randomness without some tolerance of error is not useful for classical finite-state verifiers of debates, and then prove that things change in the quantum case.

\begin{theorem}\label{theorem:derandomize}
The computational power of a public-coin probabilistic debate checking system is reduced to the level of its deterministic counterpart when the verifier is not allowed to make any error in its final decision. 
\end{theorem}
\begin{proof}
Any such zero-error probabilistic verifier  V  can be replaced by a deterministic one obtained by hard-wiring an arbitrary sequence of coin outcomes in V.
\end{proof}

As mentioned in Section \ref{section:prel}, languages with debates checkable by deterministic finite state verifiers  are regular, whereas probabilistic verifiers can handle some nonregular languages when some error is allowed. We will now see that our small quantum verifiers can do much more with zero error. 

\begin{theorem}\label{theorem:zeroerrorquantum}
Every language in the class $\mathsf{E}$ has debates checkable by  a 2qcfa with four quantum states, and with zero error.
\end{theorem}
\begin{proof}
Since $\mathsf{E}=\mathsf{ASPACE}(n)$ (the class of languages recognized by alternating Turing machines (ATMs) using linear space) \cite{CKS81}, it is sufficient to show how to trace the execution of a linear-space alternating Turing machine (ATM). Let $A$ be an ATM that decides a language $L$, using at most $kn$ tape squares for its computation on any string of length $n$, for a positive integer $k$. Assume, without loss of generality, that $A$ alternates between existential and universal states at each step, and that the start state is an existential state. 

We construct a 2qcfa V$_0$  that checks debates on membership in $L$. V$_0$ is a variant of the verifier V$_1$ described in the proof of Theorem \ref{theorem:decidable}. In this version, the debaters play a game to produce a computation history of $A$ on the input $w$ of length $n$. The protocol dictates that P1 starts by announcing the first existential choice to be made. Both debaters then transmit the start configuration of $A$ on $w$ parallelly. P0 then announces the first universal choice as a response to the first move of P1, followed by both debaters transmitting the configuration that $A$ would reach by executing the choice announced by P1 in the beginning. In general, the choice that determines configuration $c_{i+1}$ is announced by the corresponding debater before the transmission of configuration $c_i$. As usual, the verifier may order the debaters to restart the whole thing at any step.

After using it to check that the first configuration description is accurate,  V$_0$ starts moving its reading head on the input tape back and forth at the appropriate speed to make sure that neither debater sends a configuration description longer than $nk$ symbols in the rest of the transmission, deciding against any debater seen to violate this rule. As described for the verifiers in our earlier proofs, V$_0$ scans the parallel  transmissions, encoding the last configuration descriptions it has seen, as well as their legal successors according to the choices that have already been announced by the debaters. If the debaters send the same complete history, V$_0$ halts and announces the result in that history. If the debaters disagree, V$_0$ flips a coin and picks one debater's transmission to trace, just like V$_1$. Unlike V$_1$, however, V$_0$ does not trace this debater until it sends a halting configuration. Instead, V$_0$ just performs the transition check between the previously sent configuration and the presently sent one,\footnote{If the chosen debater attempts to send an exceedingly long configuration at this point, it will be caught by the control implemented by the input head.} and then issues a restart command. V$_0$ does not imitate any decision of $A$ that it may see in the transmission of the chosen debater; the only way that V$_0$ can halt without any restarts after choosing a debater is by detecting a transition error, and deciding in favor of the other debater.

If both debaters obey the protocol, then P1 will always be able to demonstrate an accepting computation history of $A$ on $w$ if $w\in L$, and P0 will always be able to demonstrate a rejecting computation history of $A$ on $w$ if $w\notin L$. So let us examine the case where one debater is lying.

If V$_0$ chooses the truthful debater to trace, it will detect no error, and so will restart with certainty. If it chooses the other debater, it will detect a transition error and announce the correct decision with some probability, and restart with the remaining probability. There is no possibility that V$_0$ can make an error. 
\end{proof}

\section{Polynomial-time debates}\label{section:poly}
Ambainis and Watrous' seminal paper \cite{AW02}, which introduced the 2qcfa model, included a demonstration of the capability of these machines to recognize the language $\{a^nb^n | n>0\}$ in polynomial expected time, a feat that is impossible for their classical counterparts. We start this section with a quick review of  the technique used there (to be called ``the AW trick" from now on) for comparing the lengths of two substrings in the input, which is  also useful for polynomial-time 2qcfa verifiers while scanning debater streams. Note that we will allow computable irrational amplitudes, as well as rationals, in the 2qcfa verifier descriptions, and  show that they outperform classical verifiers with arbitrary real transition probabilities.\footnote{2qcfa's with arbitrary real amplitudes are known to recognize languages of every Turing degree in polynomial time \cite{SY14B}.}

It is well known that the operation
	\[
		\left( 
					\begin{array}{cc}
						\cos \theta & ~-\sin \theta\\ 
						\sin \theta & ~\cos \theta
					\end{array} 
				\right),
	\]
which we denote $U_\theta$, describes a $\theta$-radian rotation of the vector representing the superposition of a single qubit with states $q_0$ and $q_1$ on the Cartesian plane with coordinate axes corresponding to  $\ket{q_0}$ and $\ket{q_1}$. To compare the lengths of a block of  $a$'s and a block of  $b$'s in a string, one  starts with the qubit at state $q_0$, corresponding to the ``undecided" outcome, and then applies $U_{\sqrt{2}\pi}$ (resp. $U_{-\sqrt{2}\pi}$) each time one reads another symbol from the $a$ (resp. $b$) block during a left-to-right scan of the string. The qubit is  measured at the end of the second block, and the process ends if one observes the state $q_1$, corresponding to the ``not-equal" outcome. 

It is certain that $q_1$ will be observed only if the lengths of the $a$ and $b$ blocks are indeed unequal, since otherwise the clockwise and counterclockwise rotations cancel out perfectly, and the qubit ends up where it started. Observing $q_0$, on the other hand, is inconclusive, since unequal lengths will cause a superposition of $\ket{q_0}$ and $\ket{q_1}$ in the qubit. Fortunately, it is shown in \cite{AW02} that the probability that $q_1$ will be measured in this case is at least $\frac{1}{2n^2}$, where $n$ is the length of the input.

Continuing our description of the AW trick, if the 2qcfa observes $q_0$ at the end of the second block, it initiates a subroutine consisting of random walks of its input head on the  tape. This subroutine runs in polynomial expected time, and  ends with the outcome ``equal" (regardless of the content of the input) with probability  $\frac{1}{4n^2}$, and the outcome ``undecided" with the remaining probability \cite{SY14A}, in which case the algorithm resets the qubit to $q_0$, and restarts the whole thing from the left end of the tape.

So if the two block lengths are equal, we will never obtain the ``not-equal" outcome, and definitely halt with the ``equal" outcome sooner or later. If they are not equal, the probability of arriving at the correct answer is at least twice that of the wrong one in every iteration, and this ratio can be improved by tuning the random walk subroutine without excessively increasing its runtime. The algorithm will halt after polynomially many restarts with high probability.

We are now ready to compare the power of classical and quantum constant-space, polynomial-time verifiers.

One limitation of the classical version is established in \cite{CHPW98}, where it is shown that any language recognized with bounded error by a 2apfa in polynomial time must have 1-tiling complexity bounded above by $2^{polylog(n)}$. The authors then go on to point out that the binary palindromes language $\mathtt{PAL}$ has 1-tiling complexity that is exponential in $n$. Although the exposition in \cite{CHPW98} is on machines which can only toss fair coins, an examination of the proofs in that paper shows that they apply to 2apfa's that use arbitrary real transition probabilities as well, and so we have
\begin{fact}
 $\mathtt{PAL}$ has no debates checkable by a polynomial-time constant-space probabilistic verifier for any error bound less than $\frac{1}{2}$.
\end{fact}

Zheng et al. \cite{ZQG05} define an interactive proof  model called QAM(2QCFA), which is easily seen to be equivalent in power to Yakary{\i}lmaz's qAM (described in Section \ref{section:prel}), and use the AW trick for the verifier to show the following
\begin{fact} 
The complement of $\mathtt{PAL}$ (i.e. the language of nonpalindromes) has a quantum Arthur-Merlin system, where the verifier uses just one qubit, and terminates with high probability in polynomial-time. 
\end{fact}
Swapping the accept and reject states of that verifier for nonpalindromes, and noting that the resulting setup can be viewed as a debate over membership in $\mathtt{PAL}$, where P1 remains silent, implicitly challenging P0 to try and prove that the input is not a palindrome, we arrive at the following result.
\begin{corollary}
The incorporation of even a single qubit to otherwise classical constant-space, polynomial-time verifiers enlarges the class of debatable languages to include $\mathtt{PAL}$.
\end{corollary}
Note that it is not even known whether there exists a nonregular language which has debates checkable by a classical constant-space,  polynomial-time verifier \cite{CHPW98}, and both $\mathtt{PAL}$ and its complement are context-free. We will now demonstrate several non-context-free languages that have debates checkable by the quantum version. Since these languages are unary, the verifier will have to depend solely on the debater stream and the input length during its execution.

\begin{theorem}\label{theorem:uprime}
The language $\mathtt{UPRIME}=\{1^p~\vert ~p\mbox{ is prime}\}$ has polynomial-time debates checkable 
by a 2qcfa with just two qubits.
\end{theorem}
\begin{proof}
The verifier V for this language is designed to ignore P1, since interaction with P0 is sufficient to establish whether the length $n$ of the input string $w$ is prime or not. P0 tries to convince V that $n$ is a composite number, i.e. a product of two integers $i,j>1$. If $n$ is indeed composite, then a string consisting of $j$ alternating blocks of the form $a^i$ and $b^i$ is a certificate of non-membership of $w$ in $\mathtt{UPRIME}$. For instance, if $w=1^{15}$, then the string $aaabbbaaabbbaaa$ is one such certificate. P0 is supposed to transmit this certificate again and again in an infinite loop. Of course, if $n$ is prime, then no such certificate exists, and all P0 can do is to send some other string and hope that V does not catch its lie.

Starting an infinite loop on the left end of the input, V moves its tape head one step to the right for each P0 symbol it reads. Making sure that no block has length one, it uses its first qubit to compare the length of each block of $a$'s in P0's transmission with the following block of $b$'s. Parallelly, the second qubit is used  to compare the length of each block of $b$'s in P0's speech with the following block of $a$'s. These comparisons are performed by applying the rotations $U_{\sqrt{2}\pi}$ and $U_{-\sqrt{2}\pi}$ on the symbols of the first and second blocks, respectively, as described in the above discussion of the AW trick. Whenever it obtains the ``not-equal" outcome on any comparison, V accepts the input, since P0 has been caught supplying a fake certificate of nonprimality. If the left-to-right scan of the input is completed without a conclusive outcome from any comparison, V calls a random walk subroutine that ends with rejection with a probability that is half of the probability that any fake certificate by P0 will cause an acceptance, and begins a new iteration and left-to-right scan otherwise. The expected runtime of V is polynomial in $n$, and the one-sided error bound (V can sometimes erroneously reject actual members of $\mathtt{UPRIME}$) can be improved by tuning the random walk subroutine, as in \cite{AW02}.
\end{proof}

\begin{theorem}\label{theorem:usquare}
The language $\mathtt{USQUARE}=\{1^{m^2}~\vert ~m>0\}$ has polynomial-time debates checkable 
by a 2qcfa with just three qubits.
\end{theorem}
\begin{proof}
The verifier V for this language is similar to the one for $\mathtt{UPRIME}$ described in the proof of Theorem \ref{theorem:uprime}, with the following differences. In this case, it is P1 who is doing all the talking, sending a certificate of ``squareness" to V over and over again. For an input string of the form $1^{m^{2}}$, such a certificate is the concatenation of exactly  $m$ alternating blocks of the form $a^m$ and $b^m$. V uses its first two qubits to make sure that all blocks are of the same length, exactly as in the previous proof. The third qubit is used to compare the number of blocks with the length of the first block, by applying a $U_{\sqrt{2}\pi}$ for each symbol of the first block, and a $U_{-\sqrt{2}\pi}$ for each block encountered in the certificate.
\end{proof}

It is easy to see how the technique used in the proof of Theorem \ref{theorem:usquare} can be generalized to handle languages involving powers greater than two, using additional qubits. Note that both $ \tt UPRIME $ and $ \tt USQUARE $ are nonstochastic \cite{Tur81}.\footnote{A language recognized by a one-way probabilistic finite automaton (pfa) with cutpoint $ \frac{1}{2} $ is called stochastic \cite{Rab63}. It is known that \cite{Kan89,Kan91,YS09C,YS11A} two-way pfa's and one-way quantum finite automata (qfa's) cannot recognize any nonstochastic language  with cutpoint $ \frac{1}{2} $. For two-way qfa's, we only know that \cite{YS09C,FYS10A,YS11A} they can recognize some non-unary nonstochastic language with cutpoint $ \frac{1}{2} $ if the head is allowed to be quantum -- a generalization of 2qcfa's.}

\begin{theorem}\label{theorem:upower}
The language $\mathtt{UPOWER}=\{1^{2^m}~\vert ~m>0\}$ has polynomial-time debates checkable 
by a 2qcfa with just two qubits.
\end{theorem}
\begin{proof}
Once again, only  P1 talks, repeating a purported certificate of membership forever. A certificate for $1^{2^m}$ is a string of length $2^m-1$, which is the concatenation of   $m$ alternating blocks of  $a$'s and $b$'s, where the $i$'th block has length $2^{i-1}$. For each transmission of the certificate, V checks whether the input is one symbol longer than the certificate, and the first block has length 1. It also uses its qubits to check that each block in the certificate is twice as long as the previous one, by applying a $U_{2\sqrt{2}\pi}$ for each symbol of the first block, and a $U_{-\sqrt{2}\pi}$ for each symbol of the second block in each comparison. The rest of the proof is similar to the previous examples.
\end{proof}

\begin{theorem}\label{theorem:ufib}
The language $\mathtt{UFIB}=\{1^{n}~\vert ~n\mbox{ is a Fibonacci number}\}$ has polynomial-time debates checkable 
by a 2qcfa with just three qubits.
\end{theorem}
\begin{proof}
The certificate of membership for the $i$'th Fibonacci number $F_i$ is a string describing the Fibonacci sequence up to $F_{i-1}$, in the format \[F_1\#F_2\#\cdots\#F_{i-3}!F_{i-2}\#F_{i-1},\] where the numbers in the sequence are given in unary, and a different separator symbol ($!$) is put before $F_{i-2}$. P1 transmits the purported certificate repeatedly. The verifier V starts by checking the first two numbers in the certificate. For all $j\in\{1,2,3\}$, and all $k\geq 3$ such that $k\equiv j-1$ (mod 3), V uses its $j$'th qubit to verify the correctness of the $k$'th member of the sequence in P1's transmission by checking whether its length equals the sum of the lengths of the two members preceding it.  V moves its  head along the input tape to make sure that the member of the certificate sequence that it is presently scanning is not longer than the input, thereby foiling any attempt to make it run forever. Whenever it sees the symbol $!$, V starts comparing the length of the input with the sum of the lengths of the two subsequent members of the certificate, using the qubit available for that comparison, and rejecting whenever it detects a mismatch. The rest of the proof is similar to the previous examples.
\end{proof}

\section{Concluding remarks}\label{section:conc}


It is well known that finite automata with $k$ classical input heads can use them as one can use logarithmic space; for instance, to count up to $O(n^k)$. One can therefore extend the argument of Theorem \ref{theorem:zeroerrorquantum} to $\mathsf{APSPACE}$ (the class of languages recognized by ATMs using polynomial space), which equals $\mathsf{EXPTIME}$ \cite{CKS81}, concluding that every language in the class $\mathsf{EXPTIME}$ has a zero-error (public-coin) debate checkable by  a multiple-head 2qcfa with four quantum states

Debate systems with deterministic logarithmic-space (or equivalently, multi-head finite-state) verifiers which have the additional property that P0 can hide some of its messages to the verifier from P1 are known to correspond to the class $\mathsf{EXPTIME}$. If one upgrades the verifier in this model to a probabilistic version, but demands that it should still make zero error, the computational power does not change, since zero-error probabilistic machines can be derandomized easily. We can therefore also state that  every language in the class $\mathsf{EXPTIME}$ has such a ``partial-information" debate checkable by  a private-coin multiple-head two-way probabilistic finite automaton with zero error.

All the examples in Section \ref{section:poly} have only one debater doing all the talking. Can one find a polynomial-time debatable language $\mathtt{L}$ where neither $\mathtt{L}$ nor its complement already have a polynomial-time qAM proof system?

\section*{Acknowledgements}
We thank the anonymous reviewers of a previous version of this manuscript for their helpful comments. 

\bibliographystyle{plain}
\bibliography{tcs}

\end{document}